\documentclass[12pt]{article}    

\usepackage[margin=0.75in]{geometry} 
\usepackage{amsmath,amssymb,amsthm}
\usepackage{color}
\usepackage{relsize}
\usepackage[perpage,symbol*]{footmisc}
\usepackage{authblk}
\usepackage{enumerate}
\usepackage{cases}

\newtheorem{theorem}{Theorem}[section]
\newtheorem{proposition}[theorem]{Proposition}
\newtheorem{lemma}[theorem]{Lemma}

\newtheorem{remark}{Remark}

\newcommand{\be}{\begin{equation}}
\newcommand{\ee}{\end{equation}}
\newcommand{\bea}{\begin{eqnarray}}
\newcommand{\eea}{\end{eqnarray}}
\newcommand{\e}{{\rm e}}

\numberwithin{equation}{section}

\begin{document}

\title{Gaussian unitary ensemble with two jump discontinuities, PDEs and the coupled Painlev\'{e} II and IV systems}

\author{Shulin Lyu\thanks{School of Mathematics (Zhuhai), Sun Yat-sen University, Zhuhai 519082, China; e-mail: lvshulin1989@163.com} ~and Yang Chen\thanks{Department of Mathematics, Faculty of Science and Technology, University of Macau, Macau, China; e-mail: yangbrookchen@yahoo.co.uk}}


\date{\today}
\maketitle
\begin{abstract}
We consider the Hankel determinant generated by the Gaussian weight with two jump discontinuities. Utilizing the results of [C. Min and Y. Chen, Math. Meth. Appl. Sci. {\bf 42} (2019), 301--321] where a second order PDE was deduced for the log derivative of the Hankel determinant by using the ladder operators adapted to orthogonal polynomials, we derive the coupled Painlev\'{e} IV system which was established in [X. Wu and S. Xu, arXiv: 2002.11240v2] by a study of the Riemann-Hilbert problem for orthogonal polynomials. Under double scaling, we show that, as $n\rightarrow\infty$, the log derivative of the Hankel determinant in the scaled variables tends to the Hamiltonian of a coupled Painlev\'{e} II system and it satisfies a second order PDE. In addition, we obtain the asymptotics for the recurrence coefficients of orthogonal polynomials, which are connected with the solutions of the coupled Painlev\'{e} II system.

\end{abstract}

$\mathbf{Keywords}$: Gaussian unitary ensembles; Hankel determinant; Painlev\'{e} equations;

Orthogonal polynomials

$\mathbf{Mathematics\:\: Subject\:\: Classification\:\: 2020}$: 33E17; 34M55; 42C05

\section{Introduction}
The $n$-dimensional Gaussian unitary ensemble (GUE for short) is a set of $n\times n$ Hermitian random matrices whose eigenvalues have the following joint probability density function
\begin{align}\label{jpdf}
p(x_1,x_2,\cdots,x_n)=\frac{1}{C_n}\cdot\frac{1}{n!}\prod_{1\leq i<j\leq n}(x_i-x_j)^2\prod_{k=1}^n \e^{-x_k^2},
\end{align}
where $x_k\in(-\infty,\infty),k=1,2,\cdots,n$.
See \cite[sections 2.5, 2.6 and 3.3]{Mehta}. The normalization constant $n!C_n$, also known as the partition function, has the following explicit representation \cite[equation (17.6.7)]{Mehta}
\begin{align*}
n!C_n:=&\int_{(-\infty,\infty)^n} \prod_{1\leq i<j\leq n}(x_i-x_j)^2\prod_{k=1}^n \e^{-x_k^2}dx_k\\
=&(2\pi)^{n/2}2^{-n^2/2}\prod_{k=1}^n k!,
\end{align*}
namely,
\[C_n=(2\pi)^{n/2}2^{-n^2/2}\prod_{k=1}^{n-1} k!.\]

We consider the Hankel determinant generated by the moments of the Gaussian weight multiplied by a factor that has two jumps, i.e.
\begin{align*}
D_n(s_1,s_2):=\det\left(\int_{-\infty}^{\infty} x^{i+j}w(x;s_1,s_2)dx\right)_{i,j=0}^{n-1},
\end{align*}
where the weight function reads
\[
w(x;s_1,s_2):=\e^{-x^2}\left(A+B_1\theta(x-s_1)+B_2\theta(x-s_2)\right), \quad x\in(-\infty,\infty),
\]
with $s_1<s_2$ and $A\geq0,A+B_1\geq0,A+B_1+B_2\geq0,B_1B_2\neq0$.
Here $\theta(x)$ is 1 for $x>0$ and 0 otherwise.
For any interval $I\subset(-\infty,\infty)$, it is well known that (see \cite[sections 2.1 and 2.2]{Szego})
\[\frac{1}{n!}\int_{I^n}\prod_{1\leq i<j\leq n}(x_i-x_j)^2\prod_{k=1}^n \e^{-x_k^2}dx_k=\det\left(\int_I x^{i+j}\e^{-x^2}dx\right)_{i,j=0}^{n-1}.\]
Therefore, the probability that the interval $(s_1,s_2)$ has all or no eigenvalues of GUE is given by $D_n(s_1,s_2)/C_n$ with $A=0,B_1=1,B_2=-1$ and $A=1,B_1=-1,B_2=1$ respectively. The former was studied in \cite{BasorChenZhang} via the ladder operator approach \cite{Ismail}, a formalism adapted to monic orthogonal polynomials, and its log derivative was shown to satisfy a two-variable generalization of the Painlev\'{e} IV system.

By using the ladder operator formalism and with the aid of four auxiliary quantities, Min and one of the authors \cite{MinChen19} derived a second order partial differntial equation (PDE for short) satisfied by
\[\sigma_n(s_1,s_2):=\left(\frac{\partial}{\partial {s_1}}+\frac{\partial}{\partial{s_2}}\right)\ln D_n(s_1,s_2).\]
In a recent paper \cite{WuXu20}, Wu and Xu studied the special case of $D_n(s_1,s_2)$ where $A=1,B_1=\omega_1-1,B_2=\omega_2-\omega_1$ with $\omega_1,\omega_2\geq0$.
 Through the Riemann-Hilbert (RH for short) formalism of orthogonal polynomials \cite{FIK}, they showed that $\sigma_n(s_1,s_2)+n(s_1+s_2)$ is the Hamiltonian of a coupled Painlev\'{e} IV system. When $s_1$ and $s_2$ tend to the soft edge of the spectrum of GUE, by applying Deift-Zhou nonlinear steepest descent analysis \cite{DeiftZhou} to the RH problem (we call it RH method below), the asymptotic formulas for $D_n(s_1,s_2)$ and the associated orthogonal polynomials were deduced, which are expressed in terms of the solution of a coupled Painlev\'{e} II system.

Comparing the finite $n$ results of the above two papers concerning $D_n(s_1,s_2)$, we ask naturally whether they are compatible with each other.
To the best knowledge of the authors, it is not easy to obtain the second order PDE of \cite{MinChen19} from the coupled Painlev\'{e} IV system of \cite{WuXu20}. What about the other side? It transpires that the Hamiltonian of the coupled Painlev\'{e} IV system of \cite{WuXu20} can be derived by using the results of \cite{MinChen19}.
This is the main purpose of the present paper, which may provide new insights into the connection between the ladder operator approach and RH problems.

As we know, the ladder operator approach  and the RH method are both very effective tools in the study of unitary ensembles. The former is elementary in the sense that it uses the very basic theory of orthogonal polynomials and it provides a quite straightforward way to derive Painlev\'{e} transcendents for finite dimensional problems particularly those involving one variable, for example, the gap probability of Gaussian and Jacobi unitary ensembles on $(-a,a)$ with $a>0$ \cite{LyuChenFan18, MinChen18}, the partition function for weight functions with Fisher-Hartwig singularities \cite{ChenHan20, MinChen20} or with potential having a pole \cite{ChenIts10}. For problems involving two variables, a second order PDE is usually deduced \cite{ChenHaqMcKay13,LyuGriffinChen19}. The RH method is powerful for asymptotic analysis, for instance, the partition function and extreme eigenvalues for weight functions with the potential having poles \cite{ACM, BMM, DaiXuZhang18}, the correlation kernel \cite{DaiXuZhang19, XuDaiZhao14}, orthogonal polynomials \cite{ChenChenFan19}. For problems involving two or several variables, a coupled Painlev\'{e} system was usually established.

This paper is built up as follows. In the next section, we present some notations and results of \cite{MinChen19}. We make use of them in section 3 to show that the four auxiliary quantities allied with the orthogonal polynomials satisfy a coupled Painlev\'{e} IV system and $\sigma_n(s_1,s_2)+n(s_1+s_2)$ is the Hamiltonian of that system. Section 4 is devoted to the discussion of the double scaling limit of the Hankel determinant. By using the finite $n$ results given in section 2, we deduce that, as $n\rightarrow\infty$, the log derivative of the Hankel determinant in the scaled variables tends to the Hamiltonian of a coupled Painlev\'{e} II system and it satisfies a second order PDE. In addition, for the recurrence coefficients of the monic orthogonal polynomials associated with $w(x;s_1,s_2)$, we obtain their asymptotic expansions in large $n$ with the coefficients of the leading order term expressed in terms of the solutions of the coupled Painlev\'{e} II system.

\section{Notations and Some Results of \cite{MinChen19}}
In this section, we present some results of \cite{MinChen19} which will be used for our later derivation in subsequent sections.

Denote the Gaussian weight by $w_0(x)$, i.e.
\[w_0(x):=\e^{-{\rm v}_0(x)}, \qquad{\rm v}_0(x)=x^2.\]
Then the weight function of our interest reads
\begin{align*}
w(x;s_1,s_2)=w_0(x)\left(A+B_1\theta(x-s_1)+B_2\theta(x-s_2)\right).
\end{align*}
It is well known that the associated Hankel determinant admits the following representation (see \cite[pp.16-19]{Ismail})
\begin{align}
D_n(s_1,s_2)=&\det\left(\int_{-\infty}^{\infty} x^{i+j}w(x;s_1,s_2)dx\right)_{i,j=0}^{n-1}\nonumber\\
=&\prod_{j=0}^{n-1}h_j(s_1,s_2).
\end{align}
Here $h_j(s_1,s_2)$ is the square of the $L^2$-norm of the $j$th-degree monic polynomial orthogonal with respect to $w(x;s_1,s_2)$, namely,
\begin{align}\label{ordef}
h_j(s_1,s_2)\delta_{jk}:=\int_{-\infty}^{\infty}P_j(x;s_1,s_2)P_k(x;s_1,s_2)w(x;s_1,s_2)dx,
\end{align}
for $j,k=0,1,2,\cdots$, and \[P_j(x;s_1,s_2):=x^j+p(j,s_1,s_2)x^{j-1}+\cdots+P_j(0;s_1,s_2).\]

From the orthogonality, there follows the three term recurrence relation
\begin{align}\label{recu}
xP_n(x;s_1,s_2)=P_{n+1}(x;s_1,s_2)+\alpha_n(s_1,s_2)P_{n}(x;s_1,s_2)+\beta_n(s_1,s_2)P_{n-1}(x;s_1,s_2)
\end{align}
with $n\geq0$, subject to the initial conditions
\[P_0(x;s_1,s_2):=1,\qquad\qquad \beta_0(s_1,s_2)P_{-1}(x;s_1,s_2):=0.\]
The recurrence coefficients are given by
\begin{align}
\alpha_n(s_1,s_2)=&p(n,s_1,s_2)-p(n+1,s_1,s_2),\label{alp}\\
\beta_n(s_1,s_2)=&\frac{h_n(s_1,s_2)}{h_{n-1}(s_1,s_2)},
\end{align}
and it follows from \eqref{alp} that
\begin{align}\label{sumal}
\sum_{j=0}^{n-1}\alpha_j(s_1,s_2)=-p(n,s_1,s_2).
\end{align}
For ease of notations, in the following discussion, we shall not display the $s_1$ and $s_2$ dependence unless necessary.

The recurrence relation implies the Christoffel-Darboux formula
\[\sum_{j=0}^{n-1}\frac{P_j(x)P_j(y)}{h_i}=\frac{P_n(x)P_{n-1}(y)-P_{n-1}(x)P_n(y)}{h_{n-1}(x-y)}.\]
Here we point out that this identity and the recurrence relation hold for general monic polynomials orthogonal with respect to any given positive function which has moments of all orders. See for example \cite[section 3.2]{Szego} for more details.

With all the above identities, one can derive a pair of ladder operators adapted to $P_n(z)=P_n(z;s_1,s_2)$:
\begin{align*}
P_n'(z)=&\beta_nA_n(z)P_{n-1}(z)-B_n(z)P_n(z),\\
P_{n-1}'(z)=&\left(B_n(z)+{\rm v}_0'(z)\right)P_{n-1}(z)-A_{n-1}(z)P_n(z),
\end{align*}
where ${\rm v}_0(z)=z^2$, $A_n(z)$ and $B_n(z)$ have simple poles at $s_1$ and $s_2$, reading
\begin{align*}
A_n(z)=&\frac{R_{n,1}(s_1,s_2)}{z-s_1}+\frac{R_{n,2}(s_1,s_2)}{z-s_2}+2,\\
B_n(z)=&\frac{r_{n,1}(s_1,s_2)}{z-s_1}+\frac{r_{n,2}(s_1,s_2)}{z-s_2},
\end{align*}
with the residues defined by
\begin{align}
R_{n,i}(s_1,s_2):=&\frac{B_iP_n^2(s_i)\e^{-s_i^2}}{h_n},\\
r_{n,i}(s_1,s_2):=&\frac{B_iP_n(s_i)P_{n-1}(s_i)\e^{-s_i^2}}{h_{n-1}}.
\end{align}
Here $P_j(s_i)=P_j(x;s_1,s_2)|_{x=s_i}$ for $j=n-1,n$. Moreover, one can show that
$A_n(z)$ and $B_n(z)$ satisfy three compatibility conditions
\begin{align}
\left(B_{n+1}(z)+B_n(z)\right)=&\left(z-\alpha_n\right)A_n(z)-{\rm v}_0'(z),\tag{$S_1$}\\
1+\left(z-\alpha_n\right)\left(B_{n+1}(z)-B_n(z)\right)=&\beta_{n+1}A_{n+1}(z)-\beta_nA_{n-1}(z),\tag{$S_2$}\\
B_n^2(z)+{\rm v}_0'(z)B_n(z)+\sum_{j=0}^{n-1}A_j(z)=&\beta_nA_n(z)A_{n-1}(z),\tag{$S_2'$}
\end{align}
where $(S_2')$ results from $(S_1)$ and $(S_2)$. See \cite[Theorem 3.1]{MinChen19}. Concerning the discussion of ladder operators and their compatibility conditions for general weight functions with jumps, refer to Lemma 1, Remark 1 and Remark 2 of \cite{BasorChen09}.

Substituting $A_n(z)$ and $B_n(z)$ into $(S_1)$ and $(S_2')$, by equating the residues on their both sides, it was found that the recurrence coefficients can be expressed in terms of the auxiliary quantities which satisfy a system of difference equations (see \cite{MinChen19}, equations (3.7)-(3.14)). The results are presented below.
\begin{proposition}
\begin{enumerate}[{\rm (a)}]
\item $R_{n,i}$ and $r_{n,i}, i=1,2,$ satisfy the following system of difference equations:
\begin{align}
\beta_nR_{n,i}R_{n-1,i}=&r_{n,i}^2,\label{Rn-1}\\
r_{n+1,i}+r_{n,i}=&(s_i-\alpha_n)R_{n,i}.\label{rn+1}
\end{align}
\item The recurrence coefficients are expressed in terms of $R_{n,i}$ and $r_{n,i}$ $(i=1,2)$ by
\begin{align}
\alpha_n=&\frac{1}{2}\left(R_{n,1}+R_{n,2}\right),\label{alR}\\
\beta_n=&\frac{1}{2}\left(r_{n,1}+r_{n,2}+n\right).\label{betar}
\end{align}
\item The quantity $\sum_{j=0}^{n-1}\left(R_{j,1}+R_{j,2}\right)$ has the following representation
\begin{align}\label{sumR}
\sum_{j=0}^{n-1}\left(R_{j,1}+R_{j,2}\right)
=-2s_1r_{n,1}-2s_2r_{n,2}+2\beta_n\left(R_{n,1}+R_{n,2}+R_{n-1,1}+R_{n-1,2}\right).
\end{align}
\end{enumerate}
\end{proposition}

By taking the derivatives of \eqref{ordef} with $j=k=n$ and $j=k+1=n$, the auxiliary quantities turn out to be the partial derivatives of $-\ln h_n(s_1,s_2)$ and $\ln p(n,s_1,s_2)$ with respect to $s_1$ and $s_2$. Refer to equations (3.15), (3.16), (3.19), (3.20) of \cite{MinChen19}. For ease of nations, in what follows, we denote $\frac{\partial}{\partial s_i}$ and $\frac{\partial^2}{\partial s_i\partial_{s_j}} (i,j=1,2)$ by $\partial_{s_i}$ and
$\partial_{s_is_j}^2$ respectively.

\begin{proposition}\label{Dhp}
The following differential relations hold
\begin{align}
\partial_{s_i}\ln h_n(s_1,s_2)=&-R_{n,i},\label{DhR}\\
\partial_{s_i}\,p(n,s_1,s_2)=&r_{n,i},\label{Dpr}
\end{align}
with $i=1,2$. In view of $\alpha_n=p(n,s_1,s_2)-p(n+1,s_1,s_2)$ and $\beta_n=h_n/h_{n-1}$, it follows that
\begin{align}
\partial_{s_i}\alpha_n(s_1,s_2)=&r_{n,i}-r_{n+1,i},\label{Dal}\\
\partial_{s_i}\beta_n(s_1,s_2)=&\beta_n\left(R_{n-1,i}-R_{n,i}\right).\label{Dbeta}
\end{align}
\end{proposition}

Define
\[\sigma_n(s_1,s_2):=\left(\partial_{s_1}+\partial_{s_2}\right)\ln D_n(s_1,s_2).\]
With the fact that $D_n(s_1,s_2)=\prod_{j=0}^{n-1}h_j(s_1,s_2)$ and by using \eqref{alR}, one finds
\begin{align}\label{sigmaR}
\sigma_n(s_1,s_2)=-\sum_{j=0}^{n-1}\left(R_{j,1}+R_{j,2}\right).
\end{align}
According to \eqref{sumal} and \eqref{DhR}, there follows
\begin{align}\label{sigmap}
\sigma_n(s_1,s_2)=2p(n,s_1,s_2),
\end{align}
so that, in light of \eqref{Dpr},
\begin{align*}
\partial_{s_i}\sigma_n=2r_{n,i},\qquad i=1,2.
\end{align*}
Hence, the compatibility condition $\partial_{s_1s_2}^2\sigma_n=\partial_{s_2s_1}^2\sigma_n$ gives us
\begin{align}\label{r1r2}
\partial_{s_2}r_{n,1}=\partial_{s_1}r_{n,2}.
\end{align}
Combining \eqref{sumR} with \eqref{sigmaR}, and taking account of \eqref{Rn-1} and \eqref{betar}, we obtain the expression of $\sigma_n(s_1,s_2)$ in terms of the auxiliary quantities
\begin{align}\label{sigmaRr}
\sigma_n=2\left(s_1r_{n,1}+s_2r_{n,2}-\frac{r_{n,1}^2}{R_{n,1}}-\frac{r_{n,2}^2}{R_{n,2}}\right)
-\left(r_{n,1}+r_{n,2}+n\right)\left(R_{n,1}+R_{n,2}\right).
\end{align}

By using the above identities, a second order PDE was established for $\sigma_n(s_1,s_2)$ (see Theorem 3.3, \cite{MinChen19}).
\begin{proposition}\label{sigmaPDE}
$\sigma_n(s_1,s_2)$ satisfies the following equation
\[\left(\left(2s_1\cdot\partial_{s_1}\sigma_n+2s_2\cdot\partial_{s_2}\sigma_n-2\sigma_n\right)^2-\Delta_1-\Delta_2\right)^2=4\Delta_1\Delta_2,\]
where $\Delta_1$ and $\Delta_2$ are defined by
\begin{align*}
\Delta_1:=&\left(\partial_{s_1s_1}^2\sigma_n+\partial_{s_2s_1}^2\sigma_n\right)^2+4\left(\partial_{s_1}\sigma_n\right)^2\left(\partial_{s_1}\sigma_n+\partial_{s_2}\sigma_n+2n\right),\\
\Delta_2:=&\left(\partial_{s_2s_2}^2\sigma_n+\partial_{s_1s_2}^2\sigma_n\right)^2+4\left(\partial_{s_2}\sigma_n\right)^2\left(\partial_{s_1}\sigma_n+\partial_{s_2}\sigma_n+2n\right).
\end{align*}
\end{proposition}

\section{PDEs satisfied by $R_{n,i}$ and Coupled Painlev\'{e} IV system}
Based on the results presented in the previous section, we will derive a coupled PDEs satisfied by $R_{n,1}$ and $R_{n,2}$ in this section, which we will see in the next section are crucial for the derivation of the coupled Painlev\'{e} II system under double scaling. We will also deduce the coupled Painlev\'{e} IV system satisfied by quantities allied with $R_{n,i}$ and $r_{n,i}$.

\subsection{Analogs of Riccati equations for $R_{n,i}$ and $r_{n,i}$, and coupled PDEs satisfied by $R_{n,i}$}

Combining the expressions involving the recurrence coefficients together, namely \eqref{alR}, \eqref{betar}, \eqref{Dal} and \eqref{Dbeta}, with the aid of the difference equations \eqref{Rn-1} and \eqref{rn+1}, we arrive at the following four first order partial differential equations for $R_{n,i}$ and $r_{n,i}$.
\begin{lemma}
The quantities $R_{n,i}$ and $r_{n,i}, i=1,2,$ satisfy the analogs of Riccati equations
\begin{align}
\partial_{s_i}\left(R_{n,1}+R_{n,2}\right)=&4r_{n,i}+(R_{n,1}+R_{n,2}-2s_i)R_{n,i},\label{Ri1}\\
\partial_{s_i}\left(r_{n,1}+r_{n,2}\right)=&\frac{2r_{n,i}^2}{R_{n,i}}-\left(n+r_{n,1}+r_{n,2}\right)R_{n,i}.\label{Ri2}
\end{align}
\end{lemma}
\begin{proof}
Removing $r_{n+1,i}$ from \eqref{Dal} by using \eqref{rn+1}, we get
\[\partial_{s_i}\alpha_n(s_1,s_2)=2r_{n,i}+\left(\alpha_n-s_i\right)R_{n,i}.\]
Inserting \eqref{alR} into the above equation, we obtain \eqref{Ri1}.

Getting rid of $R_{n-1,i}$ in \eqref{Dbeta} by using \eqref{Rn-1}, we find
\[\partial_{s_i}\beta_n=\frac{r_{n,i}^2}{R_{n,i}}-\beta_nR_{n,i}.
\]
Plugging \eqref{betar} into this identity, we come to \eqref{Ri2}.
\end{proof}

From \eqref{Ri1}, we readily get the expressions of $r_{n,i}$ in terms of $R_{n,i}$ and their first order partial derivatives. Substituting them into \eqref{Ri2}, we arrive at a coupled PDEs satisfied by $R_{n,i}$.

\begin{theorem}
The quantities $R_{n,i}, i=1,2,$ satisfy the following coupled PDEs:
\begin{subequations}\label{RPDE}
\begin{equation}\label{RPDE1}
\begin{aligned}
\left(\partial_{s_1s_1}^2+\partial_{s_1s_2}^2\right)\left(R_{n,1}+R_{n,2}\right)-\partial_{s_1}\left(R_{n,1}+R_{n,2}\right)\cdot\left(\frac{\partial_{s_1}(R_{n,1}+R_{n,2})}{2R_{n,1}}+R_{n,2}\right)+2(s_2-s_1)\left(\partial_{s_1}R_{n,2}\right)\\
+R_{n,1}\left(\partial_{s_2}(R_{n,1}+R_{n,2})-\frac{3}{2}(R_{n,1}+R_{n,2})^2+2\left(2s_1R_{n,1}+(s_1+s_2)R_{n,2}-s_1^2+2n+1\right)\right)=0,
\end{aligned}
\end{equation}
and
\begin{equation}\label{RPDE2}
\begin{aligned}
\left(\partial_{s_2s_2}^2+\partial_{s_2s_1}^2\right)\left(R_{n,1}+R_{n,2}\right)-\partial_{s_2}\left(R_{n,1}+R_{n,2}\right)\cdot\left(\frac{\partial_{s_2}(R_{n,1}+R_{n,2})}{2R_{n,2}}+R_{n,1}\right)
+2(s_1-s_2)\left(\partial_{s_2}R_{n,1}\right)\\
+R_{n,2}\left(\partial_{s_1}(R_{n,1}+R_{n,2})-\frac{3}{2}(R_{n,1}+R_{n,2})^2+2\left((s_1+s_2)R_{n,1}+2s_2R_{n,2}-s_2^2+2n+1\right)\right)=0.
\end{aligned}
\end{equation}
\end{subequations}
\end{theorem}

\begin{remark}
Interchanging $s_1$ with $s_2$, $R_{n,1}$ with $R_{n,2}$ in \eqref{RPDE1}, we get \eqref{RPDE2}. This observation agrees with the symmetry in position of $s_1$ and $s_2$ in the weight function $w(x;s_1,s_2)$ and the definitions of $R_{n,i}$.
\end{remark}

\begin{remark}
If $B_2=0$, then $R_{n,2}=0$ and $R_{n,1}$ depends only on $s_1$. Equation \eqref{RPDE1} is reduced to an ordinary differential equation satisfied by $R_{n}(s_1):=R_{n,1}(s_1,0)$
\begin{align}\label{RODE}
R_{n}''=\frac{(R_{n}')^2}{2R_n}+\frac{3}{2}R_n^3-4s_1R_n^2+2(s_1^2-2n-1)R_n,
\end{align}
which is identical with (2.37) of \cite{MinChen19} where $t_1$ is used instead of $s_1$. As was pointed out there, \eqref{RODE} can be transformed into a Painlev\'{e} IV equation satisfied by $y(s_1):=R_n(-s_1)$.

In case $B_1=0$, via a similar argument, we find that $R_{n,2}(0,s_2)$ satisfies \eqref{RODE} with $s_1$ replaced by $s_2$.
\end{remark}

\subsection{Coupled Painlev\'{e} IV system}
Define
\[
x:=\frac{s_1+s_2}{2},\qquad\qquad s:=\frac{s_2-s_1}{2},\]
and introduce four quantities allied with $R_{n,i}(s_1,s_2)$ and $r_{n,i}(s_1,s_2)$:
\begin{align*}
a_i(x,s):=&\frac{r_{n,i}^2}{R_{n,i}(r_{n,1}+r_{n,2}+n)},\\
b_i(x,s):=&\frac{R_{n,i}}{r_{n,i}}(r_{n,1}+r_{n,2}+n),
\end{align*}
with $i=1,2$. We have
\begin{align*}
s_1=x-s,\qquad\qquad s_2=x+s,
\end{align*}
and
\begin{equation}\label{rRab}
\begin{aligned}
R_{n,i}(s_1,s_2)=&\frac{a_ib_i^2}{a_1b_1+a_2b_2+n},\\
r_{n,i}(s_1,s_2)=&a_ib_i.
\end{aligned}
\end{equation}

By making use of the results from section 2, we show that $a_i$ and $b_i$ satisfy a coupled Painlev\'{e} IV system with $\left(\partial_{s_1}+\partial_{s_2}\right)\ln D_n(s_1,s_2)+n(s_1+s_2)$ being the Hamiltonian.
\begin{theorem}\label{cPiv}
The quantity \[H_{IV}(a_1,a_2,b_1,b_2;x,s):=\sigma_n(s_1,s_2)+n(s_1+s_2)\]
with $\sigma_n(s_1,s_2)=\left(\partial_{s_1}+\partial_{s_2}\right)\ln D_n(s_1,s_2)$ satisfying the second order PDE given by Proposition \ref{sigmaPDE}, is expressed in terms of $a_i(x,s)$ and $b_i(x,s)$ by
\begin{equation}\label{Hab}
\begin{aligned}
H_{IV}(a_1,a_2,b_1,b_2;x,s)=&-2(a_1b_1+a_2b_2+n)(a_1+a_2)-(a_1b_1^2+a_2b_2^2)\\
&+2\left((x-s)a_1b_1+(x+s)a_2b_2+nx\right),
\end{aligned}
\end{equation}
and it is the Hamiltonian of the following coupled Painlev\'{e} IV system
\begin{equation}\label{cPiveq}
\begin{aligned}
\partial_x a_1=\partial_{b_1}H_{IV}=&-2a_1(a_1+a_2+b_1-x+s),\\
\partial_x a_2=\partial_{b_2}H_{IV}=&-2a_2(a_1+a_2+b_2-x-s),\\
\partial_x b_1=\partial_{a_1}H_{IV}=&b_1^2+2b_1(2a_1+a_2-x+s)+2(a_2b_2+n),\\
\partial_x b_2=\partial_{a_2}H_{IV}=&b_2^2+2b_2(a_1+2a_2-x-s)+2(a_1b_1+n).
\end{aligned}
\end{equation}
\end{theorem}

Expression \eqref{Hab} follows directly from \eqref{sigmaRr} and \eqref{rRab}. To derive the coupled Painlev\'{e} IV system, we shall establish four linear equations in the variables $\partial_x a_i$ and $\partial_x b_i, i=1,2$.
Before proceeding further, we first present some results which will be used later for the derivation.

Since $\partial_x=\partial_{s_1}+\partial_{s_2}$, we readily get from \eqref{DhR} and \eqref{Dpr} that
\begin{align}
\partial_{x}\ln h_n(s_1,s_2)=&-\left(R_{n,1}+R_{n,2}\right),\label{DxhR}\\
\partial_{x}\,p(n,s_1,s_2)=&r_{n,1}+r_{n,2}.\label{Dxpr}
\end{align}
Noting that $\alpha_n=p(n,s_1,s_2)-p(n+1,s_1,s_2)$ and $\beta_n=h_n/h_{n-1}$, we find
\begin{align}
\partial_x\alpha_n(s_1,s_2)=&\sum_{i=1,2}\left(r_{n,i}-r_{n+1,i}\right),\label{Dxal}\\
\partial_{x} \beta_n(s_1,s_2)=&\beta_n\sum_{i=1,2}\left(R_{n-1,i}-R_{n,i}\right).\label{Dxbeta}
\end{align}

As an immediate consequence of \eqref{alR} and \eqref{betar}, we have
\begin{lemma}
The recurrence coefficients are expressed in terms of $a_i$ and $b_i$ by
\begin{gather}
\alpha_n=\frac{a_1b_1^2+a_2b_2^2}{2\left(a_1b_1+a_2b_2+n\right)},\label{alab}\\
\beta_n=\frac{1}{2}\left(a_1b_1+a_2b_2+n\right).\label{betab}
\end{gather}
\end{lemma}

Replacing $r_{n,1}+r_{n,2}+n$ by $2\beta_n$ in the definitions of $a_i(x,s)$, which is due to \eqref{betar}, with the aid of \eqref{Rn-1}, we build the direct relationships between $a_i$ and the quantities with index $n-1$, i.e. $R_{n-1,i}, \alpha_{n-1}$ and $ \partial_x\ln h_{n-1}$.
.
\begin{lemma} We have
\begin{equation}\label{aR}
a_i(x,s)=\frac{R_{n-1,i}}{2},\qquad i=1,2,
\end{equation}
so that, in view of \eqref{alR} and \eqref{DxhR},
\begin{align}\label{aln-1}
\alpha_{n-1}(s_1,s_2)=a_1(x,s)+a_2(x,s)=-\frac{1}{2}\,\partial_x h_{n-1}(s_1,s_2).
\end{align}
\end{lemma}

Now we are ready to deduce the four linear equations in $\partial_x a_i$ and $\partial_x b_i$, each of which will be stated as a lemma. We start from the combination of \eqref{alab} and \eqref{betab} which gives us
\begin{align}\label{eq11}
a_1b_1^2+a_2b_2^2=4\alpha_n\beta_n.
\end{align}
\begin{lemma}
We have
\begin{equation*}
\begin{aligned}
b_1^2\left(\partial_x a_1\right)+&b_2^2\left(\partial_x{a_2}\right)+2a_1b_1\left(\partial_x b_1\right)+2a_2b_2\left(\partial_x {b_2}\right)\\
=&4(a_1a_2+b_1b_2+n)(a_1a_2+b_1b_2)+2a_1b_1^2(a_1+a_2-s_1)+2a_2b_2^2(a_1+a_2-s_2).
\end{aligned}
\end{equation*}
\end{lemma}
\begin{proof}
Taking the derivative on both sides of \eqref{eq11} with respect to $x$, we have
\begin{align}\label{eq41}
\partial_x\left( a_1b_1^2+a_2b_2^2\right) =4\beta_n \left(\partial_x \alpha_n\right)+4\alpha_n \left(\partial_x \beta_n\right).
\end{align}
Now we shall make use of \eqref{Dxal} and \eqref{Dxbeta} to derive the expressions of $\partial_x\alpha_n$ and $\partial_x\beta_n$ in terms of $a_i,b_i$ or $R_{n,i}, r_{n,i}$.
Using \eqref{rn+1} to get rid of $r_{n+1,i},i=1,2,$ in \eqref{Dxal}, we find
\begin{align*}
\partial_x \alpha_n(s_1,s_2)=\sum_{i=1,2}\left(2r_{n,i}+\left(\alpha_n-s_i\right)R_{n,i}\right).
\end{align*}
On account of \eqref{aR}, we replace $R_{n-1,i}$ by $2a_i$ in \eqref{Dxbeta} and get
\begin{align*}
\partial_{x} \beta_n(s_1,s_2)=-\beta_n\left(R_{n,1}+R_{n,2}\right)+2\beta_n\left(a_1+a_2\right).
\end{align*}
Plugging the above two identities into \eqref{eq41}, we obtain
\begin{align*}
b_1^2\left(\partial_x a_1\right) +&2a_1b_1\left(\partial_x b_1\right)+b_2^2\left(\partial_x{a_2}\right)+2a_2b_2\left(\partial_x {b_2}\right)\\
=&4\beta_n\left(2(r_{n,1}+r_{n,2})-s_1R_{n,1}-s_2R_{n,2}\right)+8\alpha_n\beta_n(a_1+a_2).
\end{align*}
On substituting \eqref{rRab}, \eqref{alab} and \eqref{betab} into this equation, we come to the desired result.
\end{proof}

Replacing $n$ by $n-1$ in \eqref{Dxal} and \eqref{rn+1}, we have
\begin{align*}
\partial_x\alpha_{n-1}(s_1,s_2)=&\sum_{i=1,2}\left(r_{n-1,i}-r_{n,i}\right),\\
r_{n,i}+r_{n-1,i}=&\left(s_i-\alpha_{n-1}\right)R_{n-1,i}.
\end{align*}
Using the second equality to remove $r_{n-1,i}$ in the first one, we are led to
\[\partial_{x}\alpha_{n-1}(s_1,s_2)=\sum_{i=1,2}\left(\left(s_i-\alpha_{n-1}\right)R_{n-1,i}-2r_{n,i}\right).
\]
According to \eqref{aln-1} and \eqref{aR}, we replace $\alpha_{n-1}$ by $a_1+a_2$ and $R_{n-1,i}$ by $2a_i$ in the above identity. By taking note that $r_{n,i}=a_ib_i, i=1,2$, we come to the following equation.
\begin{lemma}
We have
\begin{align*}
\partial_x a_1+\partial_x a_2=-2a_1\left(a_1+a_2+b_1-s_1\right)-2a_2\left(a_1+a_2+b_2-s_2\right).
\end{align*}
\end{lemma}

The next equation is obtained by combining the two expressions involving $\beta_n$ and $\partial_x\beta_n$.

\begin{lemma}
We have
\begin{align*}
b_i\left(\partial_x a_i\right)+a_i\left(\partial_x b_i\right)=& 2a_i(a_1b_1+a_2b_2+n)-a_ib_i^2,\qquad i=1,2.
\end{align*}
\end{lemma}
\begin{proof}
Plugging \eqref{betar} into \eqref{Dbeta}, we get
\[\partial_{s_i}\left(r_{n,1}+r_{n,2}\right)=\left(r_{n,1}+r_{n,2}+n\right)\left(R_{n-1,i}-R_{n,i}\right),\qquad i=1,2.\]
In view of \eqref{r1r2}, i.e. $\partial_{s_2}r_{n,1}=\partial_{s_1}r_{n,2}$, we find
\begin{align*}
\left(\partial_{s_1}+\partial_{s_2}\right)r_{n,1}=\left(r_{n,1}+r_{n,2}+n\right)\left(R_{n-1,1}-R_{n,1}\right),\\ \left(\partial_{s_1}+\partial_{s_2}\right)r_{n,2}=\left(r_{n,1}+r_{n,2}+n\right)\left(R_{n-1,2}-R_{n,2}\right).
\end{align*}
Since $\partial_x=\partial_{s_1}+\partial_{s_2}$, it follows that
\begin{align*}
\partial_x r_{n,i}=\left(r_{n,1}+r_{n,2}+n\right)\left(R_{n-1,i}-R_{n,i}\right),\qquad i=1,2.
\end{align*}
Using \eqref{rRab} to replace $r_{n,i}$ and $R_{n,i}$ in this expression, and substituting $2a_i$ for $R_{n-1,i}$, which is due to \eqref{aR}, we complete the proof.
\end{proof}
\noindent
{\bf Proof of Theorem \ref{cPiv}} Now we have four linear equations in $\partial_{x}a_1,\partial_{x}a_2,\partial_{x}b_1$ and $\partial_{x}b_2$, namely,
\begin{align}
b_1^2\left(\partial_x a_1\right)+b_2^2\left(\partial_x{a_2}\right)+&2a_1b_1\left(\partial_x b_1\right)+2a_2b_2\left(\partial_x {b_2}\right)\nonumber\\
=&4(a_1a_2+b_1b_2+n)(a_1a_2+b_1b_2)+2a_1b_1^2(a_1+a_2-s_1)+2a_2b_2^2(a_1+a_2-s_2),\label{eq1}\\
\partial_x a_1+\partial_x a_2=&-2a_1\left(a_1+a_2+b_1-s_1\right)-2a_2\left(a_1+a_2+b_2-s_2\right),\label{eq2}\\
b_1\left(\partial_x a_1\right)+a_1\left(\partial_x b_1\right)=& 2a_1(a_1b_1+a_2b_2+n)-a_1b_1^2,\label{eq3}\\
b_2\left(\partial_x a_2\right)+a_2\left(\partial_x b_2\right)=& 2a_2(a_1b_1+a_2b_2+n)-a_2b_2^2.\label{eq4}
\end{align}
Subtracting \eqref{eq1} from the sum of \eqref{eq3} multiplied by $2b_1$ and \eqref{eq4} multiplied by $2b_2$, we get
\begin{align}\label{eq5}
b_1^2\left(\partial_x a_1\right)+b_2^2\left(\partial_x a_2\right)=-2a_1b_1^2\left(a_1+a_2+b_1-s_1\right)-2a_2b_2^2\left(a_1+a_2+b_2-s_2\right).
\end{align}
Combining \eqref{eq2} with \eqref{eq5} to solve for $\partial_{x}a_1$ and $\partial_{x}a_2$, and substituting the resulting expressions into \eqref{eq3} and \eqref{eq4}, we arrive at the desired coupled Painlev\'{e} IV system \eqref{cPiveq}. \quad\hfill $\square$

\begin{remark}
The Hamiltonian of the coupled Painlev\'{e} IV system presented in Theorem \ref{cPiv} is the same as the one given by (1.15) and (1.16) of \cite{WuXu20} which was derived via the Riemann-Hilbert approach. Taking note that our symbols $\beta_n$ and $h_n$ correspond to $\beta_n^2$ and $\gamma_n^{-2}$ of \cite{WuXu20}, we find that our equations \eqref{alab}, \eqref{betab} and \eqref{aln-1} are consistent with (1.23), (1.24) and (1.26) of \cite{WuXu20} respectively.
\end{remark}

\section{Coupled Painlev\'{e} II System at the Soft Edge}
We remind the reader that our weight function is obtained by multiplying the Gaussian weight by a factor with two jumps, i.e.
\[w(x;s_1,s_2)=\e^{-x^2}\left(A+B_1\theta(x-s_1)+B_2\theta(x-s_2)\right),\] where $B_1B_2\neq0$. In this section, we discuss the asymptotic behavior of the associated Hankel determinant when $s_1$ and $s_2$ tend to the soft edge of the spectrum of GUE, namely,
\[s_i:=\sqrt{2n}+\frac{t_i}{\sqrt{2}n^{1/6}}, \qquad i=1,2.\]

This double scaling may be explained in the following way. As we know, the classical Hermite polynomials $H_n(x)$ are orthogonal with respect to the Gaussian weight $\e^{-x^2}, x\in(-\infty,\infty)$. Under the double scaling $x=\sqrt{2n}+\frac{t}{\sqrt{2}n^{1/6}}$ and as $n\rightarrow\infty$, the Hermite function $\e^{-x^2/2}H_n(x)$ is approximated by the Airy function $A(x)$ multiplied by a factor involving $n$ \cite[Formula (8.22.14)]{Szego}. See also \cite[formula (3.6)]{Forrester} and \cite[Theorem 2.1]{MinChen20_1} for more explanation about this double scaling.

When $B_1=0$ or $B_2=0$, our weight function has only one jump. This case was studied in \cite{MinChen19} and the expansion formula for $R_n(s_1):=R_{n,1}(s_1,0)$ in large $n$ was given by
\[R_n(s_1)=n^{-1/6}v_1(t_1)+n^{-1/2}v_2(t_1)+n^{-5/6}v_3(t_1)+O\left(n^{-7/6}\right).\]
It was obtained by using the second order ordinary differential equation satisfied by $R_n(s_1)$ (see Theorem 2.10, \cite{MinChen19}). Hence, for our two jump case where $B_1B_2\neq0$, we assume
\begin{subequations}\label{Rnex}
\begin{equation}\label{Rn1ex}
R_{n,1}(s_1,s_2)=\sum_{i=1}^{\infty}\mu_{i}(t_1,t_2)\cdot n^{(1-2i)/6},
\end{equation}
\begin{equation}\label{Rn2ex}
R_{n,2}(s_1,s_2)=\sum_{i=1}^{\infty}\nu_{i}(t_1,t_2)\cdot n^{(1-2i)/6}.
\end{equation}
\end{subequations}
From the compatibility condition $\partial_{s_1s_2}^2R_{n,i}=\partial_{s_2s_1}^2R_{n,i},i=1,2$, it follows that
\begin{equation}\label{mnt12}
\begin{aligned}
\partial_{t_1t_2}^2\mu_{1}(t_1,t_2)=\partial_{t_2t_1}^2\mu_1(t_1,t_2),\\
\partial_{t_1t_2}^2\nu_{1}(t_1,t_2)=\partial_{t_2t_1}^2\nu_1(t_1,t_2).
\end{aligned}
\end{equation}
We keep these two relations in mind in the subsequent discussions.

Substituting \eqref{Rnex} into the left hand side of \eqref{RPDE1} and \eqref{RPDE2}, by taking their series expansions in large $n$ and setting the leading coefficients to be zero, we get a coupled PDEs satisfied by $\mu_1$ and $\nu_1$.

\begin{theorem}
The leading coefficients in the expansions of $R_{n,i}$ in large $n$, i.e.
\begin{align*}
\mu_1(t_1,t_2)=&\lim\limits_{n\rightarrow\infty}n^{1/6}R_{n,1}(s_1,s_2),\\
\nu_1(t_1,t_2)=&\lim\limits_{n\rightarrow\infty}n^{1/6} R_{n,2}(s_1,s_2),
\end{align*}
satisfy the following coupled PDEs
\begin{subequations}\label{mnuPDE}
\begin{equation}\label{mnuPDE1}
\left(\partial_{t_1t_1}^2+\partial_{t_1t_2}^2\right)(\mu_1+\nu_1)-\frac{\left(\partial_{t_1}(\mu_1+\nu_1)\right)^2}{2\mu_1}+2\mu_1(\sqrt{2}(\mu_1+\nu_1)-t_1)=0,
\end{equation}
\begin{equation}\label{mnuPDE2}
\left(\partial_{t_2t_2}^2+\partial_{t_1t_2}^2\right)(\mu_1+\nu_1)-\frac{\left(\partial_{t_2}(\mu_1+\nu_1)\right)^2}{2\nu_1}+2\nu_1(\sqrt{2}(\mu_1+\nu_1)-t_2)=0.
\end{equation}
\end{subequations}
\end{theorem}

Plugging \eqref{Rnex} into \eqref{Ri1}, we get
\begin{subequations}\label{rnex}
\begin{equation}
r_{n,1}(s_1,s_2)=\frac{\mu_1}{\sqrt{2}}n^{1/3}+\frac{\mu_2}{\sqrt{2}}+\frac{\sqrt{2}}{4}\partial_{t_1}(\mu_1+\nu_1)+O(n^{-1/3}),
\end{equation}
\begin{equation}
r_{n,2}(s_1,s_2)=\frac{\nu_1}{\sqrt{2}}n^{1/3}+\frac{\nu_2}{\sqrt{2}}+\frac{\sqrt{2}}{4}\partial_{t_2}(\mu_1+\nu_1)+O(n^{-1/3}).
\end{equation}
\end{subequations}
Hence, according to \eqref{r1r2}, i.e. $\partial_{s_2}r_{n,1}=\partial_{s_1}r_{n,2}$, we find
\begin{align}\label{Dmn}
\partial_{t_2}\mu_1(t_1,t_2)=&\partial_{t_1}\nu_1(t_1,t_2),
\end{align}

To continue, we define
\begin{align*}
v_1(t_1,t_2-t_1):=&-\frac{\mu_1(t_1,t_2)}{\sqrt{2}},\\
v_2(t_1,t_2-t_1):=&-\frac{\nu_1(t_1,t_2)}{\sqrt{2}}.
\end{align*}
With the aid of \eqref{Dmn}, we establish the following differential relations.
\begin{lemma}
We have
\begin{subequations}\label{Dmnu}
\begin{align}\label{Dm1}
\partial_{t_1}\left(\mu_1(t_1,t_2)+\nu_1(t_1,t_2)\right)=&-\sqrt{2}v_{1\xi}(t_1,t_2-t_1),
\end{align}
\begin{align}\label{Dn1}
\partial_{t_2}\left(\mu_1(t_1,t_2)+\nu_1(t_1,t_2)\right)=&-\sqrt{2}v_{2\xi}(t_1,t_2-t_1),
\end{align}
\end{subequations}
where $v_{i\xi}(i=1,2)$ denotes the first order derivative of $v_i(\xi,\eta)$ with respect to $\xi$.
\end{lemma}
\begin{proof}

By the definition of $v_1(t_1,t_2-t_1)$, we find
\begin{align*}
\partial_{t_1}\mu_1(t_1,t_2)=&-\sqrt{2}\cdot\partial_{t_1}v_1(t_1,t_2-t_1)\\
=&-\sqrt{2}\left(v_{1\xi}(t_1,t_2-t_1)-v_{1\eta}(t_1,t_2-t_1)\right),\\
\partial_{t_2}\mu_1(t_1,t_2)=&-\sqrt{2}\cdot\partial_{t_2}v_1(t_1,t_2-t_1)\\
=&-\sqrt{2}\cdot v_{1\eta}(t_1,t_2-t_1),
\end{align*}
so that
\begin{align*}
\left(\partial_{t_1}+\partial_{t_2}\right)\mu_1(t_1,t_2)=-\sqrt{2}v_{1\xi}(t_1,t_2-t_1).
\end{align*}
In view of \eqref{Dmn}, we obtain \eqref{Dm1}.
Via a similar argument, we can prove \eqref{Dn1}.
\end{proof}

With the aid of \eqref{Dmn} and \eqref{Dmnu}, we establish the following equations for $v_1$ and $v_2$ by using the coupled PDEs \eqref{mnuPDE}.

\begin{theorem}\label{vPDE}
The quantities $v_1(t_1,t_2-t_1)$ and $v_2(t_1,t_2-t_1)$ satisfy a coupled nonlinear equations
\begin{align}\label{couplePii}
v_{i\xi\xi}-\frac{v_{i\xi}^2}{2v_i}-2v_i(2(v_1+v_2)+t_i)=0,
\end{align}
where $v_{i\xi}$ and $v_{i\xi\xi}$ denote the first and second order derivative of $v_i(\xi,\eta)$ with respect to $\xi$ respectively.
\end{theorem}
\begin{proof}

Differentiation of both sides of \eqref{Dm1} over $t_1$ and $t_2$ gives us
\begin{align*}
\partial_{t_1t_1}^2\left(\mu_1(t_1,t_2)+\nu_1(t_1,t_2)\right)=&-\sqrt{2}\left(v_{1\xi\xi}(t_1,t_2-t_1)-v_{1\xi\eta}(t_1,t_2-t_1)\right),\\
\partial_{t_1t_2}^2\left(\mu_1(t_1,t_2)+\nu_1(t_1,t_2)\right)=&-\sqrt{2}v_{1\xi\eta}(t_1,t_2-t_1),
\end{align*}
where in the second equality we make use of \eqref{mnt12}. It follows that
\begin{subequations}\label{Dmn2v12}
\begin{align}
\left(\partial_{t_1t_1}^2+\partial_{t_1t_2}^2\right)\left(\mu_1(t_1,t_2)+\nu_1(t_1,t_2)\right)=-\sqrt{2}v_{1\xi\xi}(t_1,t_2-t_1).
\end{align}
Similarly, by differentiating both sides of \eqref{Dn1} over $t_1$ and $t_2$, we get
\begin{align}
\left(\partial_{t_2t_2}^2+\partial_{t_1t_2}^2\right)\left(\mu_1(t_1,t_2)+\nu_1(t_1,t_2)\right)=-\sqrt{2}v_{2\xi\xi}(t_1,t_2-t_1).
\end{align}
\end{subequations}
Plugging \eqref{Dmn2v12} and \eqref{Dmnu} into \eqref{mnuPDE}, we arrive at the desired equations.
\end{proof}

Now we look at $\sigma_n(s_1,s_2)$ which is defined by
\[\sigma_n(s_1,s_2):=\left(\partial_{s_1}+\partial_{s_2}\right)\ln D_n(s_1,s_2).\]
Recall that it is expressed in terms of $R_{n,i}$ and $r_{n,i}$ by \eqref{sigmaRr}. Substituting the expansions of $R_{n,i}$ and $r_{n,i}$ into this expression, we establish the following results.

\begin{theorem}\label{sigmaH}
$\sigma_n(s_1,s_2)$ has the following asymptotic expansion in large $n$
\begin{align}\label{sigmaHii}
\sigma_n(s_1,s_2)=\sqrt{2}n^{1/6}H_{II}(t_1,t_2-t_1)+O(n^{-1/6}),
\end{align}
where $H_{II}(t_1,t_2-t_1)$ is the Hamiltonian of the following coupled Painlev\'{e} II system
\begin{subequations}\label{Pii}
\begin{numcases}{}
v_{i\xi}=\frac{\partial H_{II}}{\partial w_i}=2v_iw_i,\label{Pii1}\\
w_{i\xi}=-\frac{\partial H_{II}}{\partial v_i}=2\left(v_1+v_2\right)+t_i-w_i^2,\label{Pii2}
\end{numcases}
\end{subequations}
which is given by
\begin{align}
H_{II}(t_1,t_2-t_1)=v_1w_1^2+v_2w_2^2-(v_1+v_2)^2-t_1v_1-t_2v_2.
\end{align}
Here $v_i=v_i(t_1,t_2-t_1)$ and $w_i=w_i(t_1,t_2-t_1)$. Moreover, $H_{II}$ satisfies the following second order second degree PDE
\begin{equation}\label{HiiPDE}
\begin{aligned}
\left(\partial_{t_1}H_{II}\right)\cdot\left(\partial_{t_2t_2}^2H_{II}+\partial_{t_2t_1}^2H_{II}\right)^2
+\left(\partial_{t_2}H_{II}\right)\cdot\left(\partial_{t_1t_1}^2 H_{II}+\partial_{t_1t_2}^2H_{II}\right)^2\\
=4\left(\partial_{t_1}H_{II}\right)\left(\partial_{t_2}H_{II}\right)\left(t_1\cdot\partial_{t_1}H_{II}+t_2\cdot\partial_{t_2}H_{II}-H_{II}\right).
\end{aligned}
\end{equation}
\end{theorem}

\begin{proof}
Recall \eqref{sigmaRr}, i.e.
\begin{align*}
\sigma_n(s_1,s_2)=2\left(s_1r_{n,1}+s_2r_{n,2}-\frac{r_{n,1}^2}{R_{n,1}}-\frac{r_{n,2}^2}{R_{n,2}}\right)
-\left(r_{n,1}+r_{n,2}+n\right)\left(R_{n,1}+R_{n,2}\right).
\end{align*}
Substituting \eqref{Rnex} and \eqref{rnex} into the right hand side of this expression, by taking its series expansion in large $n$, we obtain \begin{align*}
\sigma_n(s_1,s_2)=\left(-\frac{\left(\partial_{t_1}\left(\mu_1+\nu_1\right)\right)^2}{4\mu_1}-\frac{\left(\partial_{t_2}\left(\mu_1+\nu_1\right)\right)^2}{4\nu_1}-\frac{\left(\mu_1+\nu_1\right)^2}{\sqrt{2}}+t_1\mu_1+t_2\nu_1\right)n^{1/6}+O\left(n^{-1/6}\right).
\end{align*}
Replacing the derivative terms in the coefficient of $n^{1/6}$ by using \eqref{Dmnu}, and substituting $-\sqrt{2}v_1$ and $-\sqrt{2}v_2$  for $\mu_1$ and $\nu_1$ respectively, we find
\[\sigma_n(s_1,s_2)=\sqrt{2}n^{1/6}\left(\frac{v_{1\xi}^2}{4v_1}+\frac{v_{2\xi}^2}{4v_2}-\left(v_1+v_2\right)^2-t_1v_1-t_2v_2\right)+O\left(n^{-1/6}\right).\]
On writing
\[w_i(t_1,t_2-t_1):=\frac{v_{i\xi}(t_1,t_2-t_1)}{2v_i(t_1,t_2-t_1)},\]
we get \eqref{sigmaHii}.

From the above definition of $w_i$, we readily see that \eqref{Pii1} holds. Taking the derivative on both sides of \eqref{Pii1}, we are led to
\[v_{i\xi\xi}=4v_iw_i^2+2v_iw_{i\xi}.\]
Inserting it and \eqref{Pii1} into \eqref{couplePii}, after simplification, we produce \eqref{Pii2}.

To derive \eqref{HiiPDE}, we plugging \eqref{sigmaHii} into the PDE satisfied by $\sigma_n$, i.e. \eqref{sigmaPDE}. By taking the series expansion of its left hand side and setting the leading coefficient to be zero, we obtain \eqref{HiiPDE}.
\end{proof}

Recall \eqref{alR} and \eqref{betar} which express the recurrence coefficients in terms of $R_{n,i}$ and $r_{n,i}$, namely,
\begin{align*}
\alpha_n=&\frac{1}{2}\left(R_{n,1}+R_{n,2}\right),\\
\beta_n=&\frac{1}{2}\left(r_{n,1}+r_{n,2}+n\right).
\end{align*}
Substituting \eqref{Rnex} and \eqref{rnex} into the above expressions, after simplification, we get the asymptotic expansions of $\alpha_n$ and $\beta_n$ in large $n$.
\begin{theorem}\label{abex}
The recurrence coefficients of the monic polynomials orthogonal with respect to the Gaussian weight with two jump discontinuities have the following asymptotics for large $n$
\begin{align*}
\alpha_n(s_1,s_2)=&-\frac{v_1+v_2}{\sqrt{2}n^{1/6}}+O(n^{-1/2}),\\
\beta_n(s_1,s_2)=&\frac{n}{2}-\frac{v_1+v_2}{2}n^{1/3}+O(1).
\end{align*}
Here $v_i=v_i(t_1,t_2-t_1), i=1,2,$ satisfy the coupled Painlev\'{e} II system \eqref{Pii}.
\end{theorem}

\begin{remark}
Theorem \ref{vPDE}, \ref{sigmaH} and \ref{abex} are consistent with equation (1.33), Lemma 1 and Theorem 4 of \cite{WuXu20} respectively.
\end{remark}

\section*{Acknowledgments}
Shulin Lyu was supported by National Natural Science Foundation of China under grant number 11971492. Yang Chen was supported by the Macau Science and Technology Development Fund under grant number FDCT 023/2017/A1 and by the University of Macau under grant number MYRG 2018-00125-FST.

\end{document}